\documentclass[a4paper,onecolumn,accepted=2026-11-02]{quantumarticle}
\pdfoutput=1
\usepackage[T1]{fontenc}
\usepackage[british]{babel}
\usepackage[unicode,colorlinks]{hyperref}
\usepackage{graphicx}
\usepackage{amsmath,amssymb,amsthm,bm,relsize}
\usepackage{xspace}
\usepackage{mathtools}
\usepackage{multirow}
\usepackage{dsfont}
\usepackage{xcolor}
\usepackage{tikz,pgfplots}\pgfplotsset{compat=1.18}
\usepackage{subfig}

\usepackage{csquotes}
\usepackage[backend=bibtex,style=phys,biblabel=brackets,eprint=true,doi=false,url=true,maxnames=10,defernumbers=true,giveninits=true]{biblatex}
\bibliography{biblio}

\DeclareMathOperator{\tr}{tr}

\let\originalleft\left
\let\originalright\right
\renewcommand{\left}{\mathopen{}\mathclose\bgroup\originalleft}
\renewcommand{\right}{\aftergroup\egroup\originalright}

\newcommand{\ket}[1]{\left| #1 \right\rangle}
\newcommand{\braket}[2]{\left\langle #1 \middle| #2 \right\rangle}

\newcommand{\prin}[2]{\left\langle#1, #2\right\rangle}
\newcommand{\inner}[2]{\langle#1,#2\rangle}

\newcommand{\abs}[1]{\left|#1\right|}

\newcommand{\de}[1]{\left(#1\right)}
\newcommand{\De}[1]{\left[#1\right]}
\newcommand{\DE}[1]{\left\{#1\right\}}

\newcommand{\R}{\mathbb{R}}

\newcommand{\p}[2]{P( #1 | #2 )}

\newcommand{\pp}[2]{P(\, #1\, |\, #2\, )}

\newcommand{\fa}{\text{for all }}

\newtheorem{theorem}{Theorem}
\newtheorem*{theorem*}{Theorem}
\newtheorem{lemma}[theorem]{Lemma}
\newtheorem{definition}[theorem]{Definition}

\mathtoolsset{centercolon}

\hypersetup{pdftitle={Post-selection games}}

\begin{document}
\title{Post-selection games}
\author{Víctor Calleja Rodríguez}\affiliation{Departamento de Física Teórica, Atómica y Óptica, Universidad de Valladolid, 47011 Valladolid, Spain}
\author{Ivan A. Bocanegra-Garay}\orcid{0000-0002-5401-7778}\affiliation{Departamento de Física Teórica, Atómica y Óptica, Universidad de Valladolid, 47011 Valladolid, Spain}
\author{Mateus Araújo}\orcid{0000-0003-0155-354X}\affiliation{Departamento de Física Teórica, Atómica y Óptica, Universidad de Valladolid, 47011 Valladolid, Spain}\affiliation{Laboratory for Disruptive Interdisciplinary Science (LaDIS), Universidad de Valladolid, 47011 Valladolid, Spain.}
\date{3rd November 2026}

\begin{abstract}
	In this paper, we introduce post-selection games, a generalization of nonlocal games where each round can be not only won or lost by the players, but also discarded by the referee. Such games naturally formalize possibilistic proofs of nonlocality, such as Hardy's paradox. We develop algorithms for computing the local and Tsirelson bounds of post-selection games. Furthermore, we show that they have an unbounded advantage in statistical power over traditional nonlocal games, making them ideally suited for analysing Bell tests with low detection efficiency.
\end{abstract}
\maketitle

\section{Introduction}

The discovery that local hidden variable models cannot reproduce quantum-mechanical correlations \cite{bell1964} was a pivotal moment in the history of physics. It turned a philosophical question into an experimental one \cite{chsh1969}, and blossomed into an entire field of research \cite{brunner2014}.

The inherently probabilistic nature of the argument was widely seen as unsatisfactory, which led to attempts to prove nonlocality ``without inequalities'' \cite{gill2014}, first by Greenberger, Horne, Zeilinger, and Mermin in the multipartite case \cite{greenberger1990,mermin1990}, followed by Hardy in the bipartite case \cite{hardy1992,hardy1993nonlocality}. Nowadays, these proofs are considered to belong to different categories: GHZ-Mermin-like proofs display quantum pseudo-telepathy \cite{brassard2005b}, whereas Hardy-like proofs display possibilistic nonlocality \cite{mansfield2012}. The crucial difference is that Hardy-like proofs do not rely on some nonlocal game being won with certainty, but rather on certain events being possible or impossible.

This makes them not amenable to experimental tests, as it is not possible to demonstrate that some event is impossible; even if the frequency of some event is measured to be zero this does not imply impossibility, and in any case experimental tests of Hardy's paradox have measured nonzero frequencies for the impossible events \cite{boschi1997ladder,zhao2024,fedrizzi2011}. To obtain a conclusive demonstration of nonlocality, the latter two experiments \cite{zhao2024,fedrizzi2011} used unrelated Bell inequalities \cite{clauser1974experimental,eberhard1993background}, which are not maximally violated for correlations obeying the conditions of Hardy's paradox.

In order to solve these issues, we introduce post-selection games. They are a natural generalization of nonlocal games \cite{cleve2004} where, besides winning and losing, the players can also get a round of the game discarded. This allows us to post-select on the events entering a possibilistic proof of nonlocality, turning it into a game that can be won with certainty. This works even in the traditional Hardy scenario with two outcomes per party, where pseudo-telepathy is otherwise impossible \cite{cleve2004}.

Post-selection games do not, however, require perfect correlations to demonstrate nonlocality. Like regular pseudo-telepathy proofs, they naturally tolerate noise. Moreover, correlations that demonstrate possibilistic nonlocality win the post-selection games with maximal probability. Equivalently, the associated nonlinear Bell inequality is maximally violated by such correlations.

Regular pseudo-telepathy proofs often have much higher statistical power than probabilistic proofs of nonlocality \cite{vandam2004,araujo20}, and one might expect post-selection games to also have. To answer this question, we extend the tools of Ref.~\cite{araujo20} to measure the statistical power of both regular nonlocal games and post-selection games. We use them to compare the regular CHSH game with the Hardy post-selection game. We find, counterintuitively, that in the ideal case, when the Hardy game can be won with certainty, it is much weaker than the CHSH game. On the other hand, in the noisy case, it is much stronger, with the ratio of their statistical powers going to infinity as the noise makes the correlations local. We conjecture that this unbounded advantage holds not only against the CHSH game, but \emph{any} nonlocal game without post-selection.

\section{Preliminaries}\label{sec:hardy}

We shall have in mind a nonlocal game scenario where two separate parties, Alice and Bob, play the following cooperative game: in each round a referee sends them some inputs $x, y$ with probability $\mu(x,y)$, and they reply with some outputs $a, b$ with conditional probability $P(ab|xy)$ \cite{cleve2004}. A function $V(a,b,x,y) \in [0,1]$ then determines the probability with which they win the round. The probability of winning the game $G$ is then given by
\begin{equation}
	\omega_G(P) = \sum_{abxy} V(a,b,x,y)\mu(x,y)P(ab|xy) .
\end{equation}
It can be written in a more convenient form if we define $V_\mu(a,b,x,y) = V(a,b,x,y)\mu(x,y)$, and consider $V_\mu$ to be a vector with the same ordering as the vector $P$ of conditional probabilities $P(ab|xy)$; then
\begin{equation}
	\omega_G(P) = \inner{V_\mu}{P} ,
\end{equation}
where $\langle \cdot, \cdot\rangle$ is the natural inner product.

$P$ is known as a behaviour\footnote{Usually behaviours are defined to be nonsignalling, but this is not relevant for our purposes.} \cite{tsirelson1993}. A behaviour is called local if $P(ab|xy) = \sum_\lambda p(\lambda)p(a|x\lambda)p(b|y\lambda)$. The set of all local behaviours forms a polytope $\mathcal L$, whose vertices $\operatorname{ext}(\mathcal L)$ are deterministic behaviours. The maximal probability of winning the game with local behaviours is the local bound \cite{scarani2019bell} 
\begin{equation}
	\omega_\ell(G) := \sup_{P \in \mathcal L} \inner{V_\mu}{P} = \max_{P \in \operatorname{ext}(\mathcal L)} \inner{V_\mu}{P}.
\end{equation}
A behaviour is called quantum if $P(ab|xy) = \tr[\rho(A^a_x \otimes B^b_y)]$ for a quantum state $\rho$ and POVMs $A^a_x, B^b_y$ \cite{scarani2019bell}. The supremum of the probability of winning the game with quantum behaviours is the Tsirelson bound \cite{tsirelson80}
\begin{equation}
	\omega_q(G) := \sup_{P \in \mathcal Q} \inner{V_\mu}{P},
\end{equation}
where $\mathcal Q$ is the set of quantum behaviours. The prototypical example of a nonlocal game is the CHSH game, introduced by Tsirelson \cite{tsirelson1997}. In it $a,b,x,y \in \{0,1\}$, $\mu(x,y) = 1/4$ and $V(a,b,x,y) = 1$ if $a\oplus b = xy$ and zero otherwise. Its local bound is $3/4$, and its Tsirelson bound is $(2+\sqrt2)/4$.

\subsection{Hardy's paradox}\label{sec:oldhardy}

We can now introduce Hardy's paradox \cite{hardy1992,hardy1993nonlocality}. It consists of two steps:
\begin{enumerate}
	\item[(i)] showing that if a behaviour $P$ is local, then $P(01|01) = P(10|10) = P(00|11) = 0$ implies that $P(00|00) = 0$.
	\item[(ii)] showing that there exists a quantum behaviour $P_H$ such that $P_H(01|01) = P_H(10|10) = P_H(00|11) = 0$ but $P_H(00|00) > 0$.
\end{enumerate}
We want to formulate a nonlocal game that, in some sense, represents the paradox. A minimal desideratum is that the probability of winning this game should be maximal for behaviours $P_H$ that have the property from step (ii), called Hardy behaviours. This is, however, already impossible, because any linear Bell inequality that is maximally violated by a Hardy behaviour is necessarily also maximally violated by a local behaviour \cite{goh2018}.

Another desideratum is to make Hardy's paradox experimentally testable. For this purpose, the literature usually associates it with the CH inequality \cite{mermin1994,rabelo2012,clauser1974experimental}:
\begin{equation}
P(00|00) - P(01|01) - P(10|10) - P(00|11) \le 0.
\end{equation}
We translate it into a nonlocal game in the straightforward way, using Theorem 2 of Ref.~\cite{araujo20}, adding normalization conditions to get rid of the minus signs and re-scaling. For instance, we add the condition $P(00|01) + P(01|01) + P(10|01) + P(11|01) = 1$ to make the term $-P(01|01)$ disappear. The resulting $V$ is then given by
\begin{equation}\label{eq:vch}
	\left(\begin{array}{cc|cc}
    1 & 0 & 1 & 0 \\
    0 & 0 & 1 & 1 \\ \hline
    1 & 1 & 0 & 1 \\
    0 & 1 & 1 & 1
  \end{array}\right),
\end{equation}
where the four quadrants correspond to the inputs $x,y$ in the order ${\scriptsize \left(\begin{array}{c|c} 0,0 & 0,1 \\ \hline 1,0 & 1,1 \end{array}\right)}$. In each quadrant the four numbers correspond to the outputs $a,b$ in the order ${\scriptsize \begin{array}{cc} 0,0 & 0,1 \\ 1,0 & 1,1 \end{array}}$. We see that the resulting game is not only won in the event $(00|00)$ and lost in the events $(01|01), (10|10), (00|11)$, as we would like, but also won and lost in several irrelevant events that do not enter the paradox.

Both desiderata can be fulfilled by \emph{post-selecting} on the set of events 
\begin{equation}
    E_\text{Hardy} = \{(00|00), (01|01), (10|10), (00|11)\}.
\end{equation}
In this way, we neither win nor lose in the irrelevant events, and make the corresponding Bell inequality nonlinear, which can then be both nontrivial and maximally violated by a Hardy behaviour.

Two objections might be raised: the first is that post-selection can introduce loopholes in a Bell test, most famously the fair-sampling loophole \cite{larsson2014}. That can be handled simply by computing the local bound while taking post-selection into account, as we shall do in Section \ref{sec:local}. The second objection is that nonlinear Bell inequalities are, in general, vulnerable to memory attacks \cite{weilenmann2025}. We note that although our Bell inequalities are nonlinear, they are quasiconvex functions, and therefore the set of correlations they delimit is convex, and convex sets of correlations are not vulnerable to such memory attacks.

\section{Post-selection games}\label{sec:formal}

We can now define a post-selection game: it is a generalization of a nonlocal game defined by three functions $S, V, \mu$. As before, $\mu(x,y)$ and $V(a,b,x,y)$ determine the probability distribution over the inputs $x,y$ and the probability of winning a round with event $(ab|xy)$, respectively. $S(a,b,x,y)$ represents the probability of post-selecting a round with event $(ab|xy)$. The function $\mu$ must respect $\mu(x,y) \ge 0$ and $\sum_{xy} \mu(x,y) = 1$, while $V,S$ only need to take values in $[0,1]$, being otherwise unconstrained.

We can then represent Hardy's paradox as a post-selection game, as discussed in Section \ref{sec:oldhardy}, by letting $S(a,b,x,y) = 1$ on the set of events $E_\text{Hardy}$ and zero otherwise, and $V(a,b,x,y) = 1$ on the event $(00|00)$ and zero otherwise. Writing them in the same ordering as Equation \eqref{eq:vch}: 
\begin{equation}
  S = \left(\begin{array}{cc|cc}
    1 &  &     & 1 \\
       &  &     &  \\ \hline
       &  &  1 &  \\
    1 &  &     & 
  \end{array}\right)
  \quad\text{and}\quad
  V = \left(\begin{array}{cc|cc}
    1 &  &     &  \\
       &  &     &  \\ \hline
       &  &   &  \\
     &  &     & 
  \end{array}\right),
\end{equation}
where we omitted the entries equal to zero for clarity. We let $\mu(x,y) = 1/4$ for simplicity.

The probability of winning such a post-selection game $G$ with a behaviour $P$ is then simply the conditional probability of winning given that post-selection was successful, which is the joint probability of post-selecting and winning divided by the probability of post-selection:
\begin{equation}
	\omega_G(P) = \frac{\sum_{abxy} S(a,b,x,y)V(a,b,x,y)\mu(x,y)P(ab|xy)}{\sum_{abxy} S(a,b,x,y)\mu(x,y)P(ab|xy)}.
\end{equation}
We leave it undefined if the probability of post-selection is zero.

This can be written in a more convenient form if we define $S_\mu(a,b,x,y) = S(a,b,x,y)\mu(x,y)$, $V_\mu(a,b,x,y) = S(a,b,x,y)V(a,b,x,y)\mu(x,y)$, and regard $S_\mu, V_\mu$, and $P$ as vectors, as before:
\begin{equation}
	\omega_G(P) = \frac{\prin{V_\mu}{P}}{\prin{S_\mu}{P}}.
\end{equation}
The denominator $\inner{S_\mu}{P}$ is the probability of a successful post-selection, which we will denote $\gamma_G(P)$.

For the Hardy game, we have then
\begin{equation}
	\omega_\text{Hardy}(P) = \frac{P(00|00)}{P(00|00) + P(01|01) + P(10|10) + P(00|11)} ,
\end{equation}
and
\begin{equation}
	\gamma_\text{Hardy}(P) = \frac14\de{P(00|00) + P(01|01) + P(10|10) + P(00|11)}.
\end{equation}

Note that for any Hardy behaviour $P_H$ we have that $\omega(P_H) = 1$, so to show that it is a nontrivial post-selection game, we only need to show that the local bound is strictly smaller than one.

\subsection{Local bound}\label{sec:local}

We define the local bound of a game $G$ as
\begin{equation}
	\omega_\ell(G) := \sup_{P \in \mathcal L;\, \gamma_G(P) > 0} \omega_G(P),
\end{equation}
where $\mathcal L$ is the local polytope.

We are going to show that we can restrict our attention to the convex hull of deterministic behaviours with nonzero probability of post-selection, and then that the maximum is attained at the extreme points. This implies that, similarly to nonlocal games without postselection, we compute the local bound simply by testing a finite set of deterministic behaviours.

\begin{lemma}\label{lemma:nonzero}
For all behaviours $P$ with post-selection probability $\gamma_G(P) > 0$, there exists a behaviour $P'$ such that $\omega_G(P') = \omega_G(P)$ and $P'$ is a convex combination only of extreme points $E_i$ with post-selection probability $\gamma_G(E_i) > 0$.
\end{lemma}
\begin{proof}
Let
\begin{equation}
P := \sum_{i \in \mathcal S} \lambda_i E_i + \sum_{i \not\in \mathcal S} \lambda_i E_i,
\end{equation}
where $\mathcal S$ is the set of indices such that $\inner{S_\mu}{E_i} > 0$, and
\begin{equation}
	P' := \sum_{i \in \mathcal S} \lambda_i' E_i,
\end{equation}
where $\lambda_i' = \lambda_i/\sum_{j \in \mathcal S} \lambda_j$. Then $\omega_G(P') = \omega_G(P)$, since $\inner{S_\mu}{E_i} = 0$ implies $\inner{V_\mu}{E_i} = 0$.
\end{proof}

Now, we are going to show that the probability of winning a post-selection game is a quasiconvex function and is therefore maximized at the extreme points of its domain.
\begin{definition}
    A function $ f : C \subseteq \R^N \to \R $ is said to be quasisconvex if its domain $ C $ is convex and, for all $ \alpha \in \R $, the sublevel set
    \begin{equation}
        S_\alpha = \big\{ x \in C \,:\, f(x) \leq \alpha \big\} ,
    \end{equation}
    is convex.
\end{definition}
\begin{lemma} \label{prop:QuasiCharac}
    Let $ C \subseteq \R^N $ be a convex set. A function $ f : C \to \R $ is quasiconvex if, and only if,
    \begin{equation} \label{CuasiCarac}
        f\big( \lambda x + ( 1-\lambda ) y \big) \leq \max\big\{ f(x),\, f(y) \big\} \quad \fa x,\, y \in C \text{ and all } \lambda \in [ 0,\, 1 ].
    \end{equation}
\end{lemma}
For the proof of Lemma \ref{prop:QuasiCharac}, see section 3.4 of Ref.~\cite{boyd2004convex}.
\begin{lemma}\label{lemma:quasi}
	The probability of winning a post-selection game $\omega_G: \mathcal P \to \R$ is a quasiconvex function for any convex set of behaviours $\mathcal P$ such that $P \in \mathcal P$ implies $\gamma_G(P) > 0$.
\end{lemma}
\begin{proof}
	The sublevel sets are given by
	\begin{equation}
		S_\alpha = \DE{P \in \mathcal P\,:\, \omega_G(P) \le \alpha} .
	\end{equation}
	Since by assumption $\inner{S_\mu}{P} > 0$, we have that $\omega_G(P) \le \alpha$ iff $f(P) := \inner{V_\mu - \alpha S_\mu}{P} \le 0$. But $f(P)$ is a linear function of $P$, and in particular convex. This means that if $f(P_0) \le 0$ and $f(P_1) \le 0$ then $f(\lambda P_0 + (1-\lambda)P_1) \le 0$ for all $\lambda \in [0,1]$. This implies that for all $\alpha$ the sublevel sets $S_\alpha$ are convex.
\end{proof}
\begin{theorem}\label{thm:localbound}
\begin{equation}
	\omega_\ell(G) = \max_{E \in \operatorname{ext}(\mathcal L);\, \gamma_G(E) > 0} \omega_G(E) ,
\end{equation} where $\operatorname{ext}(\mathcal L)$ denotes the set of vertices of the local polytope.
\end{theorem}
\begin{proof}
	Because of Lemma \ref{lemma:nonzero}, the maximization can be restricted to the convex hull of the vertices with nonzero probability of post-selection. Lemmas \ref{prop:QuasiCharac} and \ref{lemma:quasi} then imply that the maximum is attained at the vertices.
\end{proof}

For the Hardy game, a computer search through the vertices finds that
\begin{equation}
	\omega_\ell(\text{Hardy}) = \frac12,
\end{equation}
showing that it is indeed a nontrivial post-selection game.

\subsection{Tsirelson bound}

We now turn to the question of how to compute the Tsirelson bound of a post-selection game. We define it as
\begin{equation}\label{eq:tsirelson}
	\omega_q(G) := \sup_{P \in \mathcal Q;\, \gamma_G(P) > 0} \omega_G(P),
\end{equation}
where $\mathcal Q$ is the set of quantum behaviours.

A simple way of computing it is to do a binary search on $\alpha \in [\omega_\ell, 1]$ by solving the feasibility problems $\inner{V_\mu}{P} \le \alpha \inner{S_\mu}{P}$ for $\inner{S_\mu}{P} > 0$ using the NPA hierarchy \cite{navascues2008}. A more elegant and efficient method can be developed, however, using a straightforward generalization of linear-fractional programming (see Section 4.3.2 of Ref.~\cite{boyd2004convex}) to arbitrary cones. 

First note that $\omega_G(\theta P) = \omega_G(P)$ for all $\theta > 0$; therefore we can rewrite problem \eqref{eq:tsirelson} as 
\begin{equation}
	\sup_{\theta > 0, P \in \mathcal Q; \prin{S_\mu}{P} > 0}  \frac{\prin{V_\mu}{\theta P}}{\prin{S_\mu}{\theta P}}.
\end{equation}
Now we use the degree of freedom $\theta$ to set the denominator to $1$, obtaining the equivalent problem
\begin{equation}
	\sup_{\theta > 0, P'/\theta \in \mathcal Q; \prin{S_\mu}{P'} = 1}  \prin{V_\mu}{P'},
\end{equation}
where we also defined $P' = \theta P$.

To turn this into a conic problem, we just need to define the convex cone of unnormalized behaviours $\mathcal{K}_\mathcal{Q}$ as $P' \in \mathcal{K}_\mathcal{Q} \Leftrightarrow \exists \theta > 0; P'/\theta \in \mathcal Q$. This is just $\mathcal Q$ without a normalization constraint, so it can be characterized by the same tools, for instance, the NPA hierarchy and therefore semidefinite programming. Our final conic problem is then
\begin{equation}\label{eq:tsirelson_conic}
	\sup_{P' \in \mathcal K_{\mathcal Q}; \prin{S_\mu}{P'} = 1}  \prin{V_\mu}{P'}.
\end{equation}
Note that there is no guarantee that the supremum will be attained, as it may be reached only in the limit of arbitrarily large $\theta$. This corresponds, in the original problem, to the supremum being reached only in the limit of arbitrarily small probability of post-selection. In such cases, the precision of a numerical solution will be limited by the size of the numbers the conic solver can work with.

As remarked before, by construction, the Hardy game has Tsirelson bound $\omega_q(\text{Hardy}) = 1$.

\section{Statistical power}

At first sight it might seem that the Hardy game is extremely powerful, in the sense of statistical power \cite{araujo20}, as the difference between its Tsirelson and local bounds is $1-1/2 = 1/2$, which is much larger than the CHSH's $(2+\sqrt{2})/4-3/4 = (\sqrt{2}-1)/4$, Magic Square's \cite{cabello2001,aravind2002} $1-8/9 = 1/9$, or GHZ-Mermin's \cite{greenberger1990,mermin1990,brassard2005b} $1-3/4=1/4$. However, this must be tempered by its low post-selection probability, which is at most $(5\sqrt{5}-11)/8$ \cite{rabelo2012} for Hardy behaviours.

To do so, we need to define what the statistical power of a post-selection game is. For that, we will use the framework of hypothesis testing, where the null hypothesis $H_0$ is the best local hidden variable theory for the post-selection game, and the alternative hypothesis $H_1$ is the quantum mechanical model for the experiment. We consider then the probability of obtaining $k$ victories in $t$ post-selected rounds conditioned on both hypotheses. For simplicity, we shall not consider the total number of rounds\footnote{We leave for future research to investigate whether it is advantageous to include it.} $n$.

Note that this differs from the notion of statistical strength defined in Ref.~\cite{vandam2004}, which considers the full statistics, not only the number of victories and rounds. The reason we deviate from it is twofold: first because using full statistics precludes considering the statistical power of specific games, which is our goal here. Second because using a simple dichotomic statistic is standard practice in the statistical analysis of Bell tests \cite{hensen2015,giustina2015,Shalm_2015}.

There are two main ways of doing the hypothesis test: using the Bayes factor or the $p$-value. We shall consider both possibilities in turn and see that we end up with the same definition of statistical power. The Bayes factor is the ratio between the probabilities of the data conditioned on either hypothesis:
\begin{equation}
	K := \frac{p(k,t|H_0)}{p(k,t|H_1)} .
\end{equation}
The Supplementary Material of Ref.~\cite{Shalm_2015} shows that the optimal local hidden variable strategy is simply playing each round independently, winning with probability equal to the local bound $\omega_\ell$. The overall probability is then
\begin{equation}
	p(k,t|H_0) = \binom{t}{k}\omega_\ell(G)^k(1-\omega_\ell(G))^{t-k} .
\end{equation}
For the alternative hypothesis, we consider that we are playing each round independently with a strategy represented by a quantum behaviour $Q$, that thus has a probability of winning $\omega_G(Q)$. We assume that $\omega_G(Q) \ge \omega_\ell(G)$, as otherwise the test is not interesting. The overall probability is then
\begin{equation}
	p(k,t|H_1) = \binom{t}{k}\omega_G(Q)^k(1-\omega_G(Q))^{t-k} ,
\end{equation}
and the Bayes factor is
\begin{equation}\label{eq:bayesdata}
	K = \frac{\omega_\ell(G)^k (1-\omega_\ell(G))^{t-k}}{\omega_G(Q)^k(1-\omega_G(Q))^{t-k}} = \exp_2\De{k \log_2 \de{\frac{\omega_\ell(G)}{\omega_G(Q)}} + (t - k)\log_2 \de{\frac{1-\omega_\ell(G)}{1-\omega_G(Q)}}}.
\end{equation}
The right-hand side is the most convenient expression for numerical computations, as it avoids underflow errors.

Now let us assume that the number of victories $k$ and the number of post-selected rounds $t$ are what we expect from behaviour $Q$, namely $k \approx t \omega_G(Q)$ and $ t \approx n \gamma_G(Q)$, where $\gamma_G(Q)$ is the post-selection probability. Then equation \eqref{eq:bayesdata} reduces to
\begin{equation}
K = \exp_2\big[-n \gamma_G(Q) D\big(\omega_G(Q) \big\| \omega_\ell(G)\big)\big],
\end{equation}
where
\begin{equation}
	D(p\|q) := p \log_2 \frac pq + (1-p) \log_2 \frac{1-p}{1-q}
\end{equation}
is the (binary) relative entropy or Kullback-Leibler divergence.

Therefore we see that to minimize the Bayes factor, and thus obtain a statistical result as conclusive as possible, we want to maximize the total number of rounds $n$ and the quantity $\gamma_G(Q)D\big( \omega_G(Q) \big\Vert \omega_\ell(G) \big)$, which we define as the \emph{statistical power} of a post-selection game $G$ with behaviour $Q$:
\begin{equation} \label{eq:statPower}
    W_G(Q) := \gamma_G(Q)D\big( \omega_G(Q) \big\Vert \omega_\ell(G) \big).
\end{equation}

Turning to the $p$-value, it is defined as the probability of achieving a result at least as \textit{extreme} as the data, given that the null hypothesis is true. For our particular case it is the probability of obtaining a number of victories that is equal or larger than the observed one, assuming that $k/t \ge \omega_\ell$. Then
\begin{equation}
	p^* := \sum_{k' = k}^t p(k',t|H_0)  = \sum_{k'=k}^t \binom{t}{k'}\omega_\ell(G)^{k'}(1-\omega_\ell(G))^{t-k'} \le \exp_2 \de{-t D\left(\frac kt \middle\| \omega_\ell(G)\right)},
\end{equation}
where the inequality comes from the Chernoff bound. Assuming, as before, that $k \approx t \omega_G(Q)$ and $ t \approx n \gamma_G(Q)$ gives us
\begin{equation}
p^* \le \exp_2\big[-n \gamma_G(Q) D\big(\omega_G(Q) \big\| \omega_\ell(G)\big)\big],
\end{equation}
the same expression as in the case of the Bayes factor, which supports the choice of Equation \eqref{eq:statPower} as the definition of statistical power.

\subsection{Examples} \label{sec:SomeNumbers}

	Let us now compare the statistical power of post-selection games and regular nonlocal games in two situations: first, in an ideal, noiseless case, and afterwards with real experimental data.

	For the CHSH game, the Tsirelson bound is $\omega_q = (2+\sqrt2)/4$, the local bound is $\omega_\ell = 3/4$, and the probability of post-selection is one. Therefore for the ideal behaviour $Q^*_C$ we have
\begin{equation}\label{eq:CmaxStatPow}
	W_\text{CHSH}(Q^*_C) = D\Bigg(\, \frac{2 + \sqrt{2}}{4} \,\Bigg\Vert\, \frac{3}{4} \,\Bigg) \approx 0.0463 .
\end{equation}
	
	For the Hardy game, the Tsirelson bound is $\omega_q = 1$, the local bound is $\omega_\ell = 1/2$, and the optimal probability of post-selection for a Hardy behaviour is $(5\sqrt{5}-11)/8$ \cite{rabelo2012}. Therefore, with this ideal behaviour $P_H^*$ the statistical power is
\begin{equation}\label{eq:HmaxStatPow}
	W_\text{Hardy}(P_H^*) = \frac{5\sqrt{5}-11}{8} D\left( 1 \,\middle\Vert\, \frac{1}{2} \right) = \frac{5\sqrt{5}-11}{8} \approx 0.0225,
\end{equation}
much lower than CHSH's.

One might think that there is no reason why a Hardy behaviour should have the maximal statistical power in the Hardy game, and that is correct. We want therefore to optimize $W_\text{Hardy}(P)$ over $P$. To do that, first notice that it is a convex function\footnote{Because it is a perspective of the relative entropy composed with a linear function. The relative entropy is a convex function, and both taking the perspective and composing with linear functions are operators that preserve convexity. See Section 3.2.6 of Ref.~\cite{boyd2004convex}.} of $P$, which implies that it is maximized at its boundary. Ref.~\cite{masanes2005extremal} shows that the boundary of the set of quantum behaviours in the scenario with 2 inputs and 2 outputs per party can be produced by projective measurements on two qubits. Therefore, we can do a simple Nelder-Mead optimization over such strategies and be reasonably confident that we found that global optimum. In this way we obtained a slightly higher statistical power, $0.02518$, which is still far from CHSH.

Instead, to beat CHSH we turn to noisy experimental data; specifically the raw counts from the loophole-free Bell test of Shalm et al. \cite{Shalm_2015}, relabelled to match our notation: 
\begin{equation}
	\left(\begin{array}{cc|cc}
    6378 &  3289 &  6794   & 2825 \\
       3147  &  44336240   &  23230 & 44311018\\ \hline
      6486 & 21358  &  106 & 27562 \\
    2818 & 44302570 &  30000   & 44274530
  \end{array}\right),
\end{equation}
where the ordering is the same as in Equation \eqref{eq:vch}.

These counts give total number of rounds $n = 177358351$, total number of victories in the CHSH game $k_C = 133027048$, total number of post-selected rounds in the Hardy game $t_H = 12127$, and total number of victories in the Hardy game $k_H = 6378$.

If we assume we had predicted the probability of victory to be equal to the frequencies\footnote{Exact equality is of course an unrealistic assumption, but the precise numbers are not relevant here. We can for example take the probabilities of victory to be instead 0.75006 and 0.52, respectively, which is about the size of the fluctuation that can be expected, and this will give Bayes factors $K_\text{CHSH} = 0.3873$ and $K_\text{Hardy} = 1.923 \times 10^{-7}$, so the conclusion is unchanged.} $k_C/n \approx 0.75005$ and $k_H/t_H \approx 0.5259$ in the CHSH and Hardy games, respectively, we can use Equation \eqref{eq:bayesdata} to compute the Bayes factors:
\begin{gather}
	K_\text{CHSH} =  0.3563, \\
	K_\text{Hardy} = 8.174 \times 10^{-8}.
\end{gather}
This shows that this data is inconclusive when analysed using the CHSH game, but a decisive rejection of local hidden variables when analysed with the Hardy game. This has been anticipated in the Supplementary Material of Ref.~\cite{Shalm_2015}, where they did the statistical analysis of a loophole-free Bell test. However, they viewed it as a purely statistical technique being applied to the traditional CHSH game. By reinterpreting it as a nonlocal game itself, we not only make it conceptually simpler, but also easy to optimize the statistical test.

In order to make the numbers comparable to the ones in Equations (\ref{eq:CmaxStatPow}) and (\ref{eq:HmaxStatPow}) we compute\footnote{Which is equivalent to pretending that these frequencies are probabilities and computing the statistical power directly using Equation \eqref{eq:statPower}.} $-\frac1n\log_2 K$, obtaining $8.399 \times 10^{-9}$ and $1.327 \times 10^{-7}$. We see that even though both are very small, their ratio is roughly 16, which means that one can get a conclusion with the same statistical robustness using 16 times fewer rounds with the Hardy game.

In the next section, we will show analytically that this ratio can get arbitrarily large.

\subsection{Unbounded advantage}

Let us consider a semi-idealized photonic experiment where the only imperfection is limited detection efficiency. That is, whenever a photon arrives at a detector, it is detected with probability $\eta$ and lost with probability with $1-\eta$. In order not to introduce a loophole, lost photons are binned together with the outcome $1$, which is an optimal strategy in the scenario with 2 inputs and 2 outputs per party \cite{branciard2011}.

The detection probabilities with efficiency taken into account form then an effective behaviour $E_\eta(P)$, which is a linear function of the ideal behaviour $P$ defined as
\begin{subequations} \label{eqs:effexpectedbehav}
    \begin{align}
        & E_\eta(\p{00}{xy}) = \eta^2\p{00}{xy}, \\
        & E_\eta(\p{01}{xy}) = \eta\p{01}{xy} + \eta(1-\eta)\p{00}{xy}, \\
        & E_\eta(\p{10}{xy}) = \eta\p{10}{xy} + \eta(1-\eta)\p{00}{xy}, \\
        & E_\eta(\p{11}{xy}) = \p{11}{xy} + (1-\eta)\left(\p{01}{xy} +\p{10}{xy}\right) + (1-\eta)^2\p{00}{xy}.
    \end{align}
\end{subequations}

We want to compute the statistical power of the CHSH and Hardy games for a given (ideal) quantum behaviour $Q$ and efficiency $\eta$, namely
\begin{gather}
	W_\text{CHSH}\big( E_\eta(Q) \big) = D\bigg(\omega_\text{CHSH}\big(E_\eta(Q)\big)\bigg\Vert \frac{3}{4} \bigg), \label{eq:chshpower} \\
	W_\text{Hardy}\big( E_\eta(Q) \big) = \gamma_\text{Hardy}\big(E_\eta(Q)\big)D\bigg(\omega_\text{Hardy}\big(E_\eta(Q)\big)\bigg\Vert\frac12\bigg). \label{eq:hardypower}
\end{gather}
We want to show that the ratio between them can get arbitrarily large. More precisely, we want to show that for any given lower bound there exist $\eta$ and $Q$ such that
\begin{equation}\label{eq:ratio}
	R_{\eta,Q} = \frac{W_\text{Hardy}\big( E_\eta(Q) \big)}{\sup_{Q' \in \mathcal Q} W_\text{CHSH}\big( E_\eta(Q') \big)} ,
\end{equation}
is larger than the lower bound.

Let us start with the denominator. Maximizing $W_\text{CHSH}\big( E_\eta(Q) \big)$ boils down to maximizing $\omega_\text{CHSH}\big(E_\eta(Q)\big)$. This problem was solved in Ref.~\cite{gigena2024robust}. Adapting to our notation we have that for $\eta \in (2/3,1]$
\begin{equation}
	\sup_{Q \in \mathcal Q} \omega_\text{CHSH}\big(E_\eta(Q)\big) = \frac{ \eta^2r(\eta) + 2(1-\eta)^2 + 4 }{ 8 },
\end{equation}
where $r(\eta)$ is the largest real root of the following degree-four polynomial in $\lambda$:
\begin{multline}\label{eq:polyn}
    f_s(\lambda) = \lambda^4 + (-s^2 + 4) \lambda^3 + \left(\frac{11}{4} s^4 - 12 s^2 - 4\right) \lambda^2 \\
    + (2 s^6 - s^4 - 20 s^2 - 32) \lambda + 5 s^6 - 21 s^4 + 16 s^2 - 32,
\end{multline}
where $s = 2(1-\eta)/\eta \in [0,1)$.

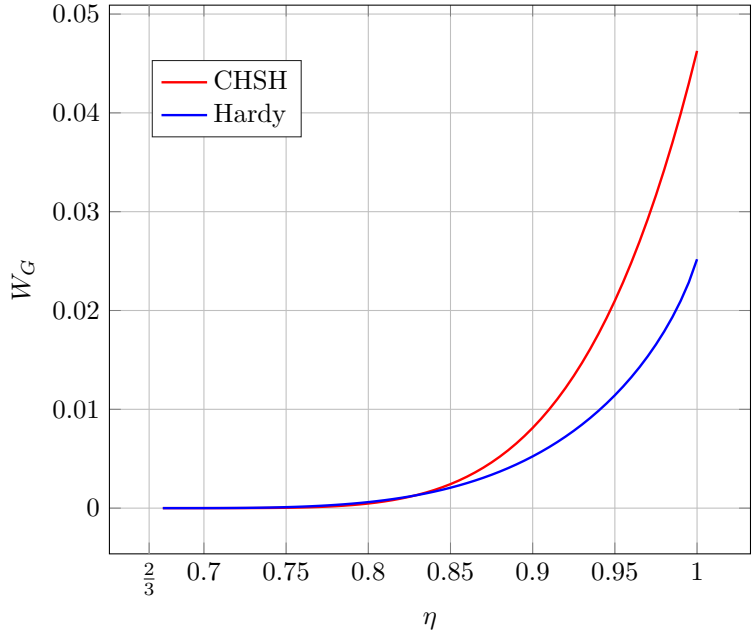
\begin{figure}[htb]
	\centering
 	\begin{tikzpicture}
		\begin{axis}[%
            yticklabel style={
                /pgf/number format/fixed,
                /pgf/number format/precision=5
            },
            scaled y ticks=false,
			scale only axis,
            extra x ticks={0.66666},
            extra x tick labels={$\frac{2}{3}$},
            xtick={0.7, 0.75, 0.8, 0.85, 0.9, 0.95, 1},
			grid=major,
			xlabel={$\eta$},
			ylabel = {$W_G$},
			axis background/.style={fill=white},
			legend style={at={(0.3,0.9)},legend cell align=left, align=left, draw=white!15!black}
			]
            \addplot[color=red, line width=0.9pt] table[col sep=space] {plot_data/relent_chsh};
        	\addlegendentry{CHSH}
            \addplot[color=blue, line width=0.9pt] table[col sep=space] {plot_data/relent_hardy};
        	\addlegendentry{Hardy}   
        \end{axis}
	\end{tikzpicture}
	\caption{Maximal statistical power of CHSH and Hardy games as a function of detection efficiency, Equations (\ref{eq:chshpower})-(\ref{eq:hardypower}).}
 \label{fig:both_relent}
\end{figure}

\begin{figure}[htb]
	\centering
 	\begin{tikzpicture}
		\begin{axis}[%
			scale only axis,
            extra x ticks={0.66666},
            extra x tick labels={$\frac{2}{3}$},
            xtick={0.7, 0.75, 0.8, 0.85, 0.9, 0.95, 1},
			grid=major,
			xlabel={$\eta$},
			ylabel = {$R_{\eta,\tilde Q_\eta}$},
			axis background/.style={fill=white},
			legend style={at={(0.94,0.3)},legend cell align=left, align=left, draw=white!15!black}
			]
            \addplot[smooth, color=red, line width=0.9pt] table[col sep=space] {plot_data/ratio};
        \end{axis}
	\end{tikzpicture}
	\caption{Ratio between maximal statistical powers of CHSH and Hardy games as a function of detection efficiency, Equation \eqref{eq:ratio}.}
 \label{fig:ratio}
\end{figure}
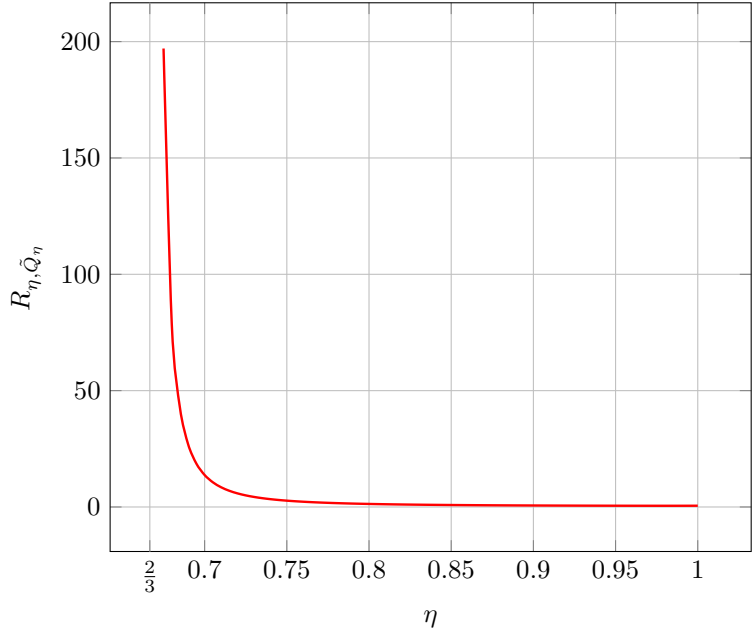
For the numerator, we maximize $W_\text{Hardy}\big( E_\eta(Q) \big)$ using the same strategy as in Section \ref{sec:SomeNumbers}, obtaining a behaviour $\tilde Q_\eta$ that is likely optimal. In Figure \ref{fig:both_relent} we plot both $W_\text{Hardy}\big( E_\eta(\tilde Q_\eta) \big)$ and $\sup_{Q'} W_\text{CHSH}\big( E_\eta(Q') \big)$ separately. We can see that although CHSH's power is much larger at $\eta = 1$, it gets lower than Hardy's around $\eta = 0.83$. In Figure \ref{fig:ratio} we plot $R_{\eta,\tilde Q_\eta}$, which seems to diverge as $\eta \to 2/3$. We are going to prove that this is in fact the case.

In order to do that, we need to find a family of behaviours $Q_\eta$ that gives us an analytical expression for the statistical power of the Hardy game. We will use the form described by Hardy in Ref.~\cite{hardy1993nonlocality}, which was shown in Ref.~\cite{hwang1996} to be capable of demonstrating nonlocality for all $\eta > 2/3$. We consider the two-qubit state
\begin{equation}\label{eq:state}
\ket{\psi} = N\Big[ \alpha\beta\big( \ket{01} + \ket{10} \big) + \beta^2\ket{11} \Big] ,
\end{equation}
with $ \alpha, \beta \geq 0 $, $ \alpha^2 + \beta^2 = 1 $, and $ N = 1/(2\alpha^2\beta^2 + \beta^4)^{1/2} = 1/( 1 - \alpha^4 )^{1/2} $. We also consider two projective measurements for each party with eigenvectors
\begin{gather}
\ket{A_0^0} = \ket{B_0^0} = \beta\ket{0} - \alpha\ket{1}, \qquad \ket{A_0^1} = \ket{B_0^1} = \alpha\ket{0} + \beta\ket{1}, \label{eq:obs1} \\
\ket{A_1^0} = \ket{B_1^0} = \ket{0}, \qquad\ket{A_1^1} = \ket{B_1^1} = \ket{1}.\label{eq:obs2}
\end{gather}
Then, we define $ Q_\eta(ab|xy) = \abs{ \braket{A_x^a B_y^b}{\psi} }^2 $, which is a Hardy behaviour when $ 0 < \alpha < 1 $. We define $ z = \alpha^2 $ and we find
\begin{equation}
    \gamma_\text{Hardy}\big(E_\eta(Q_\eta)\big) = \frac{ \eta z^2 (2-\eta-\eta z)}{4(1+z)} ,
\end{equation}
and
\begin{equation}
    \omega_\text{Hardy}\big(E_\eta(Q_\eta)\big) = \frac{\eta(1-z)}{2-\eta- \eta z }.
\end{equation}
We need that $ \omega_\text{Hardy}\big(E_\eta(Q_\eta)\big) > 1/2 $ to have nonlocality, which implies $ z < 3 - 2 / \eta $. Maximizing $ \gamma_\text{Hardy}\big(E_1(Q_\eta)\big)$ over $z$, we find the maximum at $ z_0 = \frac{1}{2}\big(\sqrt{5}-1\big)$. To make $z$ a function of $\eta$ we do a linear interpolation between $0$ at $\eta=2/3$ and $z_0$ at $\eta=1$, obtaining
\begin{equation} \label{eq:linearz}
    z(\eta) = 3z_0(\eta - 2/3),
\end{equation}
which in fact satisfies the condition $ z(\eta) < 3 - 2 / \eta $ for $\eta > 2/3$.

Let then
\begin{gather}
	w_H(\eta) = W_\text{Hardy}\big(E_\eta(Q_\eta)\big) , \\
	w_C(\eta) = \sup_Q W_\text{CHSH}\big(E_\eta(Q)\big),
\end{gather}
such that
\begin{equation}
	R_{\eta, Q_\eta} = \frac{w_H(\eta)}{w_C(\eta)}.
\end{equation}
We want to show that
\begin{equation}\label{eq:infinitepower}
	\lim_{\eta \to {2/3}^+} R_{\eta, Q_\eta} = \infty .
\end{equation}
Since $w_H(2/3) = w_C(2/3) = 0$, we use L'Hôpital's rule to compute the limit. Computing the derivatives of $w_H(\eta)$ is straightforward, the first four are
\begin{equation}
	w_H'(2/3) = w_H''(2/3) = w_H'''(2/3) = 0\quad\text{and}\quad w_H''''(2/3) = \frac{2781-1215\sqrt5}{4\log(2)}
\end{equation}
To compute the derivatives of $w_C(\eta)$ we use implicit differentiation, as the explicit expression is too complicated even for the computer to work with. It's straightforward but tedious, so we use a computer algebra system to compute the first six derivatives:
\begin{equation}
	w_C'(2/3) = w_C''(2/3) = w_C'''(2/3) = w_C''''(2/3) = w_C'''''(2/3) = 0\quad\text{and}\quad w_C''''''(2/3) = \frac{87480}{\log(2)}.
\end{equation}
Since the fourth derivative of $w_C(\eta)$ is zero but the fourth derivative of $w_H(\eta)$ is larger than zero, a recursive application of L'Hôpital's rule gives us the desired limit.

We conjecture that this unbounded advantage is not only over the CHSH game, but over \emph{any} nonlocal game without post-selection. More precisely, we conjecture that
	\begin{equation}
	 \forall\eta\quad	\sup_{Q\in\mathcal Q, G \in \mathcal G} W_G(E_\eta(Q)) = \sup_{Q \in \mathcal Q} W_\text{CHSH}(E_\eta(Q))
	\end{equation}
where $\mathcal G$ is the set of regular nonlocal games.

To support this conjecture, we ran an unstrucured search over all $2^{16}$ nonlocal games with deterministic $V$ and uniform $\mu$ for several values of $\eta$, and found that all of them had statistical power smaller than CHSH's (modulo relabellings). We also did a local optimization with the Nelder-Mead algorithm of the statistical power over the coefficients of the CHSH game for several values of $\eta$, and found that it is at a local optimum. For these calculations we used the libraries Ket.jl \cite{ket} and Optim.jl \cite{mogensen2018}.

\subsection{Non-possibilistic example}

	Nontrivial post-selection games with no connection to possibilistic nonlocality can also be generated. We performed an unstructured search over the\footnote{They are fewer than the $4^{16}$ one might naïvely expect because we exclude the games that have $V(a,b,x,y) > 0$ when $S(a,b,x,y) = 0$, as those are redundant. The total number is then $\sum_{k=0}^{16} \binom{16}{k}2^k = 3^{16}$.} $3^{16}$ post-selection games in the 2222 scenario with deterministic $V,S$ and uniform $\mu$, searching for games that were still nontrivial for efficiency $\eta$ close to $2/3$. There were only two (modulo relabellings): the Hardy game itself and the game $G_\neg$ with
\begin{equation}
    S_{G_\neg} = \left(\begin{array}{cc|cc}
    1 & 1 &    1 &  \\
     1  &  &     &  \\ \hline
    1   &  &  1 &  \\
     &  &     & 
  \end{array}\right)
  \quad\text{and}\quad
  V_{G_\neg} = \left(\begin{array}{cc|cc}
     &  &    1 &  \\
       &  &     &  \\ \hline
     1  &  &   &  \\
     &  &     & 
  \end{array}\right).
\end{equation}
Using Theorem \ref{thm:localbound}, we compute its local bound to be
\begin{equation}
	\omega_\ell(G_\neg) = \frac12,
\end{equation}
and using the conic problem \eqref{eq:tsirelson_conic}, we compute an approximate upper bound to its Tsirelson bound
\begin{equation}
	\omega_q(G_\neg) \lessapprox 0.666.
\end{equation}
Evaluating the strategy given by Equations (\ref{eq:state}--\ref{eq:obs2}) with this game, we find that
\begin{equation}
	\omega_{G_\neg}(\ket{\psi},\ket{A_x^a},\ket{B_y^b}) = \frac{2N^2\alpha^4\beta^2}{4N^2\alpha^4\beta^2 - N^2\alpha^4\beta^4} = \frac{2}{4-\beta^2}
\end{equation}
which is equal to $2/3$ in the limit $\beta \to 1$, raising the possibility that this is indeed the Tsirelson bound.

To prove that this is the case, we'll show that the no-signalling bound is also equal to $2/3$. First we write out explicitly the probability of winning the game:
\begin{equation}
	\omega_{G_\neg}(P) = \frac{P(00|01) + P(00|10)}{P(00|01) + P(00|10) + P(00|11) + P(00|00) + P(01|00) + P(10|00)}
\end{equation}
We then use the chain of inequalities
\begin{equation}
	P(00|00) + P(01|00) + P(10|00) \ge \frac{P_A(0|0) + P_B(0|0)}2 \ge \frac{P(00|01) + P(00|10)}2,
\end{equation}
which implicitly use no-signalling, together with $P(00|11) \ge 0$, to conclude that
\begin{equation}
	\omega_{G_\neg}(P) \stackrel{NS}{\le} \frac{P(00|01) + P(00|10)}{P(00|01) + P(00|10) + \frac{P(00|01) + P(00|10)}2} \le \frac23.
\end{equation}

Unlike the Hardy game, the Tsirelson bound is only reached in the limit of the post-selection probability going to zero, so the statistical power there is zero. In order to find out its maximal statistical power, we need to optimize it over $P$ for $\omega_{G_\neg}(P) \in [1/2, 2/3]$. To do that we note that for fixed $\omega_{G_\neg}(P)$ the statistical power is maximized simply by maximizing the post-selection probability, which we can do by adapting the conic problem \eqref{eq:tsirelson_conic}. The result, which we confirm by finding an explicit strategy reaching it, is that the maximal statistical power is $0.00617$, obtained with $\omega_{G_\neg}(P) \approx 0.5931$ and $\gamma_{G_\neg}(P) \approx 0.2451$.

This statistical power is much lower than Hardy's, but one might hope that it can become larger for lower efficiencies, as happened in the case of Hardy versus CHSH. Alas, this is not the case; we computed the statistical power of the Hardy game and $G_\neg$ for several values of $\eta \in (2/3, 1]$ using the same technique as above, and found that Hardy's was always larger.

\section{Generalizations of the Hardy game}

Several authors have proposed different generalizations of Hardy's paradox \cite{boschi1997ladder,cabello2002,liang2003,liang2005,chen2013hardy,cereceda2017,meng2018hardy,chen2024}, with error tolerance and more inputs and outputs. Here we shall focus on the generalizations from Refs.~\cite{boschi1997ladder,chen2013hardy,meng2018hardy}, as they can be straightforwardly reformulated as post-selection games. We adopt the formulation of Meng et al. \cite{meng2018hardy}, as it can recover the others as particular cases. Let then Alice and Bob have $s$ inputs and $k$ outputs each, and define $\p{a{<}b}{xy} := \sum_{a<b}P(ab|xy)$. The generalized Hardy paradox is as follows: if a behaviour $P$ is local, then
\begin{subequations} \label{HardyCondits3}
\begin{alignat}{2}
& \pp{ a{<}b }{ 0, s-1 } = 0, \label{condit1} \\
& \pp{ a{<}b }{ i, {i-1} } = 0 \qquad & \forall i \in \{ 1, \ldots, s-1\}, \label{condit2}\\
    & \pp{ a{>}b }{ i-1, i-1 } = 0 \qquad & \forall i \in \{ 1, \ldots, s-1\}. \label{condit3}
\end{alignat}
\end{subequations}
imply that
\begin{equation} \label{HardyConseq3}
    \pp{ a{<}b }{ s-1, s-1 } = 0.
\end{equation}
However, there exists a quantum behaviour that satisfies conditions \eqref{HardyCondits3} but violates equation \eqref{HardyConseq3}.

We define the generalized Hardy game as post-selecting on the events entering conditions (\ref{HardyCondits3})-(\ref{HardyConseq3}), winning if the event appearing in condition \eqref{HardyConseq3} occurs and losing otherwise. In order to avoid needlessly discarding rounds of the game, we define the probability distribution $\mu(x,y)$ to be uniform over the $2s$ inputs appearing in conditions (\ref{HardyCondits3})-(\ref{HardyConseq3}) and zero otherwise.

The probability of winning the game with a behaviour $P$ given that post-selection was succesful is then 
\begin{equation} \label{WinPost3}
    \omega(P) = \frac{ \pp{ a{<}b }{ s-1, s-1 } }{ \pp{ a{<}b }{ s-1, s-1 } + \pp{ a{<}b }{ 0, s-1 } + \sum_{i=1}^{s-1}\big( \pp{ a{<}b }{ i, i-1 } + \pp{ a{>}b }{ i-1, i-1 } \big) }.
\end{equation}
By construction, $\omega(P_H) = 1$ for any behaviour $P_H$ that respects conditions \eqref{HardyCondits3} and violates condition \eqref{HardyConseq3}, which we call again a Hardy behaviour. This implies that the Tsirelson bound is $\omega_q = 1$.

In order to compute the local bound, note that the proof of the generalized Hardy's paradox implies that $\omega_\ell < 1$. Note furthermore that the numerator of Equation \eqref{WinPost3} is the probability of a single event, and therefore can be only either 0 or 1 for deterministic behaviours. Since Theorem \ref{thm:localbound} implies that the local bound is attained by deterministic behaviours, we have that $\omega_\ell \le \frac12$. This bound can be achieved for example with a behaviour where Bob always answers 1, and Alice always answers 0, unless her input is 0, in which cases she answers 1. Therefore the local bound is $\omega_\ell = \frac12$.

The statistical power of the generalized Hardy game for any Hardy behaviour $P_H$ is then given by \eqref{eq:statPower}:
\begin{equation}\label{eq:hardypower_general}
	W_{\text{Hardy}_{s,k}}(P_H) = \gamma(P_H)D\left( 1 \middle\Vert \frac{1}{2} \right) = \pp{ a{<}b }{ s-1, s-1 }/2s.
\end{equation}
Optimizing it reduces to optimizing the so-called Hardy probability $\pp{ a{<}b }{ s-1, s-1 }$, which we do with the techniques from \cite{boschi1997ladder, meng2018hardy}, using two $ k $-dimensional qudits and projective measurements. The results are shown in Table \ref{tab:mqHardy}. We see that the statistical power peaks at $s=3$, but always increases with increasing $k$.

\begin{table}[t]
    \centering
    $$
    \begin{array}{c|ccccccccc} 
        \hline\hline
        & \; s = 2 & & s = 3 & & s = 4 & & s = 5 & & s = 6 \\ \hline
        k=2\;\; &\;  0.0225 & & 0.0291 & & 0.0289 & & 0.0271 & & 0.0250 \\
	k=3\;\; & \; 0.0353 & & 0.0446 & & 0.0435 & & 0.0402 & & 0.0366 \\
        k=4\;\; & \; 0.0441 & & 0.0549 & & 0.0528 & & 0.0483 & & 0.0437 \\
        k=5\;\; & \; 0.0508 & & 0.0624 & & 0.0594 & & 0.0539 & & 0.0485 \\
        k=6\;\; & \; 0.0560 & & 0.0682 & & 0.0645 & & 0.0581 & & 0.0521 \\ \hline\hline
    \end{array}
    $$
    \caption{Lower bounds on the statistical power \eqref{eq:hardypower_general} of the generalized Hardy game with $s$ inputs and $k$ outputs per party.}
    \label{tab:mqHardy}
\end{table}





\section{Code availability}

Implementations of the algorithms for computing the local bound and the Tsirelson bound introduced in Section \ref{sec:formal} are available in the following repository: \url{https://github.com/araujoms/Post-selection-games}.

\section{Author contributions}

MA came up with the main idea, proved Theorem \ref{thm:localbound} and formulated the conic problem \eqref{eq:tsirelson_conic}. VCR proved equation \eqref{eq:infinitepower} and defined the generalized Hardy games. VCR and IABG wrote the software and performed the numerical calculations. All authors contributed to the discussions and writing the manuscript.

\section{Acknowledgements}

We thank an anonymous referee for suggesting the proof of the Tsirelson bound of the $G_\neg$ game.

The research of IABG and MA was supported by the Q-CAYLE project, funded by the European Union-Next Generation UE/MICIU/Plan de Recuperación, Transformación y Resiliencia/Junta de Castilla y León (PRTRC17.11), and also by the Department of Education of the Junta de Castilla y León and FEDER Funds (Reference: CLU-2023-1-05). MA was also supported by the Spanish Agencia Estatal de Investigación, Grants No. RYC2023-044074-I and PID2024-161725OA-I00

\printbibliography

\end{document}